\documentclass[fleqn,letterpaper,UKenglish,11pt]{article}
\usepackage[margin=1in]{geometry}

\usepackage[l2tabu, orthodox]{nag}
\usepackage[utf8]{inputenc}
\usepackage[pdfusetitle]{hyperref}
\usepackage{microtype}
\usepackage{tcolorbox}

\usepackage{xcolor}

\usepackage{mathtools}

\usepackage{amsmath,amssymb,amsthm}

\theoremstyle{plain}

\newtheorem{theorem}{Theorem}[section]

\newtheorem*{conjecture*}{Conjecture}

\newtheorem{lemma}[theorem]{Lemma}

\newtheorem{corollary}[theorem]{Corollary}
\theoremstyle{definition}
\newtheorem{definition}[theorem]{Definition}

\newcommand{\Fix}[2]{\operatorname{Fix}(#1,#2)}
\newcommand{\lQ}{(l,Q)}

\usepackage{enumerate} %Gives ability, e.g., \begin{enumerate}[Case
\usepackage{authblk}

\usepackage{xspace}

\usepackage{todonotes}

\usepackage{booktabs}

\newcommand{\Sig}{\operatorname{Sig}}
\newcommand{\Tup}{\operatorname{Tup}}
\newcommand{\Cl}{\operatorname{Cl}}
\newcommand{\pr}{\operatorname{pr}}
\newcommand{\Reach}{\textnormal{\textsc{Reach}}}
\newcommand{\NextOne}{\textnormal{\textsc{Next-one}}}
\newcommand{\Next}{\textnormal{\textsc{Next}}}
\newcommand{\Solve}{\textnormal{\textsc{Solve}}}

\newcommand{\size}[1]{\lVert #1\rVert}
\DeclareMathOperator{\CSP}{CSP}
\DeclareMathOperator{\Inv}{Inv}
\newcommand{\Inst}{\mathcal{I}}

\newenvironment{procedure}[2]{\par\medskip\noindent\textbf{Procedure} \textsc{#1}$(#2)$\par\nopagebreak\noindent}{\par\medskip}
\newcounter{stepi}
\newcounter{stepii}
\newenvironment{steps}{%
  \begin{list}{\textbf{Step \arabic{stepi}.}}{\usecounter{stepi}%
    \setlength{\labelwidth}{3.4em}\setlength{\labelsep}{0.5em}\setlength{\leftmargin}{3.9em}%
    \setlength{\itemsep}{1pt}\setlength{\parsep}{0pt}\setlength{\topsep}{2pt}}}{\end{list}}

\newcounter{romi}

\begin{document}
\title{A Fast Algorithm for Maltsev Constraints}
 \author[1]{Victor Lagerkvist}
 \affil[1]{Link\"oping University, Sweden}

\date{}
\maketitle

\begin{abstract}
The \emph{constraint satisfaction problem} over a set of relations $\Gamma$ ($\CSP(\Gamma)$) is the computational problem of deciding if a set of constraints admits at least one solution. The classical complexity for finite-domain $\CSP(\Gamma)$ is settled by the \emph{CSP dichotomy theorem}: it is tractable if $\Gamma$ satisfies a non-trivial algebraic invariant and is NP-complete otherwise. However, not all these algebraic invariants result in efficient algorithms despite being theoretically tractable. A notable case that generalizes linear equations is that of \emph{Maltsev} CSPs: an $n$-variable instance with $m$ constraints is solvable in roughly $O(n^8 \cdot m)$ time by Bulatov \& Dalmau (SIAM J.\ Comput.\ 2006) or $O(n^4 \cdot m)$ time by Dyer \& Richerby (SIAM J.\ Comput.\ 2013). At the same time, arguably, most ``natural'' and efficiently usable polynomial-time algorithms rarely exceed a quadratic or cubic time bound. In this paper we revisit Maltsev constraints with this question in mind and find a $O(n^2 \cdot m)$ algorithm (for finite languages, for infinite languages we in addition need to take the total size of the instance into account). The main novel idea is to not attempt to improve the bottleneck in Bulatov \& Dalmau (the \textsc{Fix-Values} procedure) but to avoid it altogether with a slightly more refined approach that allows us to search through a smaller space.
\end{abstract}

\section{Introduction}
The \emph{constraint satisfaction problem} over a set of relations $\Gamma$ ($\CSP(\Gamma)$) is the well-known computational problem of deciding if a set of constraints over $\Gamma$ admits at least one satisfying assignment. For finite domains the classical complexity is by now well understood: $\CSP(\Gamma)$ is tractable if $\Gamma$ satisfies a non-trivial algebraic condition~\cite{Zhuk:jacm2020,Bulatov:focs2017}, and is NP-hard otherwise, by a gadget reduction from 3-SAT~\cite{BJK05,BOP18}. There are, roughly speaking, two broad and natural classes of tractable CSPs: those solvable by local consistency methods, and those solvable by (algebraic generalizations of) Gaussian elimination~\cite{FV98}. While the algorithms of Bulatov~\cite{Bulatov:focs2017} and Zhuk~\cite{Zhuk:jacm2020} handle certain involved combinations of local consistency and Gaussian elimination, the algorithms are complicated and are mainly intended as a proof of concept. Whether there is a simpler, unified algorithm is a topic of intense research for both CSP and its promise version~\cite{BBKO21,BGWZ20,CZ23,DO24}.
 
The limits and applications of simple but powerful local consistency algorithms such as (singleton) arc consistency are well understood~\cite{BK14,FV98,Kozik21}. The same is true for the reach of the algebraic generalization of linear equations, in the sense that the CSPs it solves are exactly characterized~\cite{BIMMVW10,BD06,IMMVW10}. What is lacking on the latter side is efficiency: the known algorithms are polynomial but far from practical, and, unlike local consistency, are rarely used in practice.
Towards defining the algebraic generalization of linear equations, first note that a system of linear equations (mod $d$) in $n$ variables satisfies the following algebraic invariant: for any three solutions $s,t,u \in [d]^n$ the tuple $s - u + t$ (again, mod $d$) satisfies this system, too. This can be generalized to (coset) equations over a group $(D, \cdot)$ where the corresponding operation guaranteed to preserve the solution space is $x \cdot y^{-1} \cdot z$. If one abstracts this further then the relevant property is a ternary operation $\varphi \colon D^3 \to D$ such that $\varphi(x,x,y) = \varphi(y,x,x) = y$ for any $x,y \in D$. Such an operation is called a \emph{Maltsev} operation and if we require that each relation in $\Gamma$ is ``preserved'' by $\varphi$ (see Section~\ref{sec:prelim} for details) then we get a reasonable algebraic generalization of linear equations (\emph{Maltsev CSPs}).

The celebrated result by Bulatov \& Dalmau~\cite{BD06}, the \emph{simple algorithm for Maltsev constraints}, then solves every Maltsev CSP in polynomial time. In the sequel, let us abbreviate this algorithm by \textsc{BD}. It is conceptually simple and is built upon the idea that every relation preserved by a Maltsev operation has a small generating set (a \emph{compact representation}). We then go through the constraints $C_1, \ldots, C_m$ in the instance and for each $C_l$, $1 \leq l \leq m$, assess its impact upon the solution space by constructing a compact representation of the solution space of $(C_1, \ldots, C_l)$. We then either detect an inconsistency (no solution can be represented) or output a compact representation of the solution space of $C_1, \ldots, C_m$. The main complication is to compute these compact representations efficiently. To see exactly where the complexity bottleneck in \textsc{BD} lies we thus need to define a few technical notions.
For $R\subseteq D^{n}$, a pair of tuples $t,t'\in R$ \emph{witnesses} a triple $(i,a,b)$ if $t$ and $t'$ agree on the first $i-1$ coordinates and $t_{i}=a$, $t'_{i}=b$. The pair $(a,b)$ is sometimes called a \emph{fork} of $R$ at coordinate $i$~\cite{BMS19}.
The \emph{signature} $\Sig(R)$ is the set of witnessed triples, so $|\Sig(R)|\le n|D|^{2}$ however large $R$ is.
A \emph{representation} of $R$ is a subset $R'\subseteq R$ with $\Sig(R')=\Sig(R)$, and it is \emph{compact} if $|R'|\le 2|\Sig(R)|$. The key fact is that a representation of a relation preserved by $\varphi$ generates it under $\varphi$~\cite[Lemma~3.1]{BD06}.
\textsc{BD} then uses a procedure \textsc{Solve} that maintains a compact representation of the solution space $R_l$ of the first $l$ constraints $C_1, \ldots, C_l$.
To update the current compact representation \textsc{BD} uses the procedure \textsc{Next}, which adds the constraint scope one coordinate at a time by calling the procedure \textsc{Next-beta}.
\textsc{Next-beta} must find, for every $(i,a,b)$ in the signature of the new relation, two witnesses with a common prefix, and accomplishes this with two subroutines:
\textsc{Nonempty} finds a tuple of $R$ whose projection onto given coordinates lies in a given relation, by closing the representation under $\varphi$, and
\textsc{Fix-values} takes a representation of $R$ and values $a_{1},\dots,a_{i-1}$ and returns a compact representation of the tuples of $R$ whose first $i-1$ coordinates are $a_{1},\dots,a_{i-1}$. The second witness is found by \textsc{Nonempty} inside that representation.
Thus, \textsc{Fix-values} is called once for each of the $n|D|^{2}$ triples $(i,a,b)$ (once per $(i,a)$ if its result is saved), and each call fixes up to $n$ coordinates, calling \textsc{Nonempty} $n|D|^{2}$ times per coordinate.
Bulatov and Dalmau do not make the polynomial explicit but Dalmau's technical report~\cite{Dal04} bounds one coordinate step (with respect to a relation $S$) by $O(n^{8}+(n+|S|)^{4}|S|n^{2})$ and the algorithm by $O(mn^{8}+m(n+|S^{*}|)^{4}|S^{*}|n^{2})$ for a finite language with largest (in the sense of cardinality) relation $S^{*}$.
 
The \textsc{BD} algorithm and various questions concerning Maltsev CSPs (and the related \emph{subpower membership} problem) have been revisited several times in the literature in different contexts~\cite{Dal06,IMMVW10,BIMMVW10,BM25,Mayr12,BMS19,Kom24,Gay25,BD07,BDGM12,DH17,DS26,Car16,CL17}. 
Notably, Dyer and Richerby~\cite{DR13}, in the context of simplifying the dichotomy for finite-domain counting CSPs, report an improvement of \textsc{BD}.
They replace compact representations by (small) \emph{frames} that describe (at most $n(|D|-1)+1$ entries) tuples with a common prefix.
Similar to a compact representation a frame generates $R$ through applications of the Maltsev operation~\cite[Lemma~15]{DR13}.
Their update step avoids the worst redundancies of \textsc{Nonempty} (a closure procedure that considers each triple once, \cite[Lemma~13]{DR13}) and fixes a prefix in $O(n^{3})$ time \cite[Corollary~17]{DR13}. However, it still computes, for every prefix class at every coordinate, a frame of the tuples with that prefix, and then intersects with the constraint in question.
For a fixed language and constant $|D|$ this gives $O(n^{4})$ per constraint and $O(mn^{4})$ in total~\cite[Lemmas~19--20]{DR13}.
To the best of our knowledge this is still the best explicit bound for Maltsev CSPs. Furthermore, it runs into the exact same wall as \textsc{BD}:
for each of about $n|D|$ pairs (coordinate, value or class), one build a representation of the tuples of $R$ with the corresponding prefix, and search inside it.
A representation of a prefix class has up to $n(|D|-1)+1$ (for a frame) or $2n|D|^{2}$ (for a compact representation) tuples of length $n$, so merely writing it down can cost $\Theta(n^{2})$. The cheapest known way of computing it costs $O(n^{3})$~\cite{DR13,Kom24}, which over $n|D|$ classes gives $O(n^{4})$ per constraint.

\paragraph{Our paper.}
In this paper we break the $O(n^{4}\cdot m)$ barrier with an obvious (in hindsight) but novel approach.
Rather than attempting to speed up \textsc{Fix-values} we avoid calling it in the first place.
We first investigate what \textsc{Next-beta} actually extracts from the tuples of $R$ with a given prefix, i.e., which values occur at the one position $i$ under consideration, and which of these tuples are compatible with the new constraint. Both questions depend only on the coordinates in position $i$ and in the scope $Q$ of the constraint. So all that is needed is the \emph{projection} of the prefix class onto these $|Q|+1$ coordinates. This is a set of at most $|D|^{2}|S|$ tuples (thanks to the coordinate-at-a-time idea of \textsc{BD}).
 
The question is how to get at this projection without going through the entire class. This turn out to be straightforward.
The first property is a localised form of the generating lemma (i.e., that a relation can be generated from its compact representation with the Maltsev operation). Fix a tuple $u$ of $R$ with respect to some prefix, pick $(a,a')$ from the current compact representation (witnessing a fork at some position $\geq i$), and consider \emph{moves} of the form $x\mapsto\varphi(x,a,a')$. 
There are roughly $n|D|^{2}$ such maps. Because $a$ and $a'$ agree before their position, a move never changes the first $i-1$ coordinates of $x$, and we do not leave the prefix class. Because of the two identities defining a Maltsev operation, the moves reach every tuple of the class from $u$. Thus the class is the set of tuples reachable from $u$ along $n|D|^{2}$ specific applications of the Maltsev operation (that, importantly, follows immediately from the compact representation).
Second, we apply these ``moves'' to projections (with respect to the current constraint and a fixed index).
Hence, the reachability search can be done with breadth-first search on at most $|D|^{2}|S|$ vertices with $n|D|^{2}$ moves each. 
Whenever the search discovers a new projected tuple, we, as output, also store a full tuple of the class projecting to it, since this later is necessary to build the new representation.
The same search, started from any tuple of $R$ and using all the moves, then yields the projection of all of $R$. 
 
One coordinate step therefore costs $n$ positions times a search that spends $O(n|D|^{2})$ per object, i.e.\ $O(n^{2}|D|^{5}|S|)$.
Summing over the constraints gives
$O\bigl(|D|^{5}\,n^{2}\,\size{\Inst}\bigr)$ where $\size{\Inst}$ is the size of the input. For a fixed, finite constraint language we obtain $O(mn^{2})$ which matches Gaussian elimination.

This work naturally opens up corresponding questions for \emph{$k$-edge terms} (of which Maltsev is the special case $k = 2$) where the induced CSP is known to be solvable in $O(m(n|D||S^{*}|)^{(k+2)^{2}})$ time~\cite{IMMVW10}. Several steps in our algorithm carries over and the exponent can plausibly be lowered. However, whether the exponent can actually reach $O(k)$ is an intriguing open question.

\section{Preliminaries} \label{sec:prelim}

Throughout we assume that $D$ is a finite set of values (the \emph{domain}) with $d \geq 2$ elements.

\subsection{Tuples, relations, and projections}
For $t = (t_1, \ldots, t_n)) \in D^{n}$ and $l \in [n]$ write $t[l]$ or $t_{l}$ for the $l$th coordinate of $t$ and $t[{<}l]=(t_{1},\dots,t_{l-1})\in D^{l-1}$ for the prefix induced by the first $l - 1$ elements. If $l = 1$ then we get the empty tuple $\varepsilon$. We write $t_{>1}=(t_{2},\dots,t_{s})\in D^{s-1}$ for the tuple with the first entry removed. 

An \emph{index sequence} is a sequence $I=(i_{1},\dots,i_{s})\in[n]^{s}$, $s\ge0$, where  repetition is allowed. The \emph{projection} of a tuple $t\in D^{n}$ with respect to $I$ is $\pr_{I}(t)=(t_{i_{1}},\dots,t_{i_{s}})\in D^{s}$, and for $X\subseteq D^{n}$, $\pr_{I}(X)=\{\pr_{I}(t) \mid t\in X\}$. 
We write $\pr_{l}$ for $\pr_{(l)}$, and for $Y\subseteq D^{s}$ and $0\le j\le s$, $\pr_{\le j}(Y)=\{(y_{1},\dots,y_{j}):y\in Y\}$ (so $\pr_{\le0}(Y)=\{\varepsilon\}$ for $Y\ne\emptyset$).
For $Q=(q_{1},\dots,q_{k})$ we put $Q_{\le j}=(q_{1},\dots,q_{j})$, and for index sequences $I,J$ we write $(I,J)$ for their concatenation. In this context we make no difference between  $l\in[n]$ and $(l)$. Hence, $(l,Q)=(l,q_{1},\dots,q_{k})\in[n]^{k+1}$.

\subsection{The constraint satisfaction problem and invariants}

For finite $D$ a \emph{constraint satisfaction problem} over $D$ is given by a set of variables $V = \{x_1, \ldots, x_n\}$ and a set of \emph{constraints} $(Q, S)$ where $S \subseteq D^k$ is a $k$-ary relation (for some $k \geq 1$) and $Q \in [n]^k$ a tuple of variable indices (the \emph{scope}). The question is whether there exists at least one satisfying assignment, that is, a function $f \colon V \to D$ such that $(f(x_{i_1}, \ldots, f(x_{i_k}))) \in S$ for each constraint $((i_1, \ldots, i_k),S)$. We often represent assignments by tuples and $t \in D^n$ then satisfies the instance if $\pr_Q(t) \in S$ for each constraint $(Q,S)$.

In this paper we are interested in a subclass of CSPs obtained by fixing a certain algebraic property.
Let $R \subseteq D^n$ be an $n$-ary relation over $D$. Then $R$ is said to be \emph{closed}, or \emph{invariant}, under a $k$-ary $f \colon D^k \to D$ if 
\[f(t^1, \ldots, t^k) \in R\]
for each sequence $t^1, \ldots, t^k \in R$ (applied component-wise). The function $f$ is also called a \emph{polymorphism} of $R$. For a fixed operation $f \colon D^k \to D$ one can then consider the CSP problem induced by $f$ ($\CSP(\Inv(f))$) where we for each constraint $(Q, S)$ require that $S$ is closed under $f$. We use the standard assumption that relations are represented by explicitly listing the tuples of relations.

The operation of special interest in this paper is a ternary operation $\varphi \colon D^3 \to D$ that for all $x,y \in D$ satisfy $\varphi(x,x,y) = \varphi(y,x,x) = y$. Such an operation is said to be a \emph{Maltsev} operation. We call a relation invariant under a Maltsev operation a \emph{Maltsev closed} relation. For a fixed Maltsev operation $\varphi \colon D^3 \to D$ we then want to solve the problem $\CSP(\Inv(\varphi))$ (``Maltsev CSP'').

Let $S \subseteq D^k$ be a Maltsev closed, $k$-ary relation, in an $n$-variable Maltsev CSP instance.
Then we remark that for a constraint $(Q,S)$, $Q\in[n]^{k}$, its \emph{solution set} 
\[
  C_{Q,S}=\{t\in D^{n} \mid \pr_Q(t) \in S\}
\]
is also Maltsev closed.

\subsection{Signatures and representations}
We borrow the following standard notions from Bulatov \& Dalmau~\cite{BD06}.

\begin{definition}[signature]
For $R\subseteq D^{n}$,
\[
  \Sig(R)=\bigl\{(l,b,c)\in[n]\times D\times D:\ \exists\,t,t'\in R\text{ with }t[{<}l]=t'[{<}l],\ t[l]=b,\ t'[l]=c\bigr\}.
\]
A pair $(t,t')$ as in the definition \emph{witnesses} $(l,b,c)$ and is sometimes called a \emph{fork}.
\end{definition}
Note that $(l,b,b)\in\Sig(R)$ if and only if $b\in\pr_{l}(R)$, and that $|\Sig(R)|\le n d^{2}$.

\begin{definition}[representation]\label{def:rep}
A \emph{representation} of $R\subseteq D^{n}$ is a finite list $W$ of triples $(l,a,a')$ with $l\in[n]$, $a,a'\in R$ and $a[{<}l]=a'[{<}l]$, such that for every $(l,b,c)\in\Sig(R)$ there is $(l,a,a')\in W$ with $a[l]=b$ and $a'[l]=c$.
We call $l$ the \emph{position} of the triple and write
\[
  \Tup(W)=\{a,a':(l,a,a')\in W\},\qquad W_{\ge l}=\{(l',a,a')\in W: l'\ge l\}.
\]
A triple $(l,a,a')\in W$ \emph{witnesses} the signature $(l,a[l],a'[l])$.
\end{definition}

 A representation is \emph{compact} if no two of its triples witness the same signature, so that $|W|\le|\Sig(R)|\le n d^{2}$.
In Bulatov \& Dalmau~\cite{BD06} a (compact) representation of $R$ is a subset $R'\subseteq R$ with $\Sig(R')=\Sig(R)$ (and $|R'|\le 2|\Sig(R)|$).
The two notions are interchangeable: if $W$ is a representation in our sense then $\Tup(W)$ is one in theirs, and from $R'$ one obtains $W$ by choosing one witnessing pair in $R'$ for each signature.
We keep the pairs explicit because it slightly simplifies the presentation of the algorithm. 

\section{The Algorithm}

Throughout this section we assume that $\varphi$ is a Maltsev operation over a fixed $D$. We introduce basic notation used in the algorithm and prove some elementary properties in Section~\ref{sec:lemmas}, and present the actual algorithm in Section~\ref{sec:algorithm}. We wrap up with a correctness proof and a complexity analysis in Section~\ref{sec:correctness}.

\subsection{Basic notions and properties}\label{sec:lemmas}

For a relation $R\subseteq D^{n}$, $l\in[n]$ and $p\in D^{l-1}$ let
\[
  \Fix{R}{p}=\{t\in R \mid  t[{<}l]=p\}
\]
be the set of tuples of $R$ whose first $l-1$ coordinates are fixed to $p$.
For $l=1$, $\Fix{R}{\varepsilon}=R$.
We observe that if $R$ is invariant under  $\varphi$ then any fixation $\Fix{R}{p}$ is invariant under $\varphi$, too (this follows from the fact that $\varphi(x,x,x) = x$ for every $x \in D$, i.e., it is idempotent).

\begin{definition} \label{def:move}
For $a,a'\in D^{n}$ the \emph{move} $f_{a,a'}\colon D^{n}\to D^{n}$ is $f_{a,a'}(x)=\varphi(x,a,a')$.
If $I$ is an index sequence of length $s \leq n$ and $a,a'\in D^{n}$, the move $f^{I}_{a,a'}:=f_{\pr_{I}(a),\pr_{I}(a')}$ on $D^{s}$ is called the \emph{projected move} of $(a,a')$ along~$I$.
\end{definition}

For a set $F \subseteq X^X$ of maps on a set $X$ and $u\in X$ we write $\Cl_{F}(u)$ for  the smallest subset of $X$ that contains $u$ and is closed under every $f\in F$, i.e., if $x \in \Cl_F(u)$ and $f \in F$ then $f(x) \in \Cl_F(u)$.
For a set of triples $V \subseteq [n] \times D^n \times D^n$ we write $F(V)=\{f_{a,a'} \mid (l', a,a') \in V\}$ and, for an index sequence $I$, $F_{I}(V)=\{f_{a_{I},a'_{I}}:(l', a, a')\in V\}$ for the set of projected moves of $V$ along $I$. This allows us to use a set of triples as a set of maps whenever convenient.

The aim of the remainder of this section is now to recover parts of the compact representation machinery in Bulatov \& Dalmau~\cite{BD06}, somewhat specialized to our setting since we later (Section~\ref{sec:algorithm}) apply it differently. 

\begin{lemma}\label{lem:move}
Let $R\subseteq D^{n}$ be Maltsev closed, $l'\in[n]$, and $a,a'\in R$ with $a[{<}l']=a'[{<}l']$.
Then for every $x\in R$:
\begin{itemize}
\item $f_{a,a'}(x)\in R$,
\item $f_{a,a'}(x)_{i}=x_{i}$ for all $i<l'$,
\item if $x_{l'}=a_{l'}$ then $f_{a,a'}(x)_{l'}=a'_{l'}$.
\end{itemize}
\end{lemma}
\begin{proof}
The first statement holds because $R$ is Maltsev closed.
For the second statement, for $i<l'$, $a_{i}=a'_{i}$, so $f_{a,a'}(x)_{i}=\varphi(x_{i},a_{i},a_{i})=x_{i}$.
Last, $f_{a,a'}(x)_{l'}=\varphi(a_{l'},a_{l'},a'_{l'})=a'_{l'}$.
\end{proof}

Next, we obtain a simple but useful characterization of forks of the kind $(l,b,c)$ for fixed $l$ and $b$.

\begin{lemma} \label{lem:transport}
Let $T\subseteq D^{n}$ be Maltsev closed.
\begin{itemize}
\item If $(l,b,c)\in\Sig(T)$ and $u\in T$ with $u[l]=b$, then there is $t'\in T$ with $t'[{<}l]=u[{<}l]$ and $t'[l]=c$.
\item For all $l\in[n]$, $b\in D$ and $u\in T$ with $u[l]=b$,
\begin{equation}\label{eq:transport}
  \{c\in D \mid (l,b,c)\in\Sig(T)\}=\pr_{l}\bigl(T_{u[<l]}\bigr).
\end{equation}
\end{itemize}
\end{lemma}
\begin{proof}
First, let $s,s'\in T$ witness $(l,b,c)$ and put $t'=f_{s,s'}(u)$. Lemma~\ref{lem:move} then gives $t'\in T$, $t'[{<}l]=u[{<}l]$ and $t'[l]=c$.

Second, fix $l$, $b$ and $u$.
For the $\subseteq$ direction: if $(l,b,c)\in\Sig(T)$, then by the first claim there is $t'\in T_{u[<l]}$ with $t'[l]=c$, so $c\in\pr_{l}(T_{u[<l]})$.
For the $\supseteq$ direction: if $c\in\pr_{l}(T_{u[<l]})$, pick $t'\in T$ with $t'[{<}l]=u[{<}l]$ and $t'[l]=c$; then $(u,t')$ witnesses $(l,b,c)\in\Sig(T)$.
\end{proof}

To foreshadow the algorithm in Section~\ref{sec:algorithm}, assume that we are given a Maltsev closed $n$-ary relation $R$ with a representation $W$, an index $1 \leq l \leq n$, a prefix $p\in D^{l-1}$, and let $u\in\Fix{R}{p}$ be any tuple of $R$ with this prefix.
We will prove that $\Fix{R}{p}$ can be obtained from $u$ by repeated applications of the form $\varphi(x,a,a')$, where $x$ is a tuple obtained so far, and where $(l',a,a')$ is a triple of $W$ with $l'\ge l$.
That a closed relation is generated by a compact representation is the basic fact behind the algorithms of Bulatov and Dalmau~\cite{BD06,Dal04} and the generalisation to few subpowers~\cite{IMMVW10}. It is also used for the subpower membership problem~\cite{Mayr12,BMS19,Kom24}.
Two things are different in our setting. 
First, we do not explicitly close a set under the ternary operation $\varphi$, instead we apply to a single tuple only the unary maps $x\mapsto\varphi(x,a,a')$, where $(a,a')$ is taken from $W$.
Second, we use only the triples of $W$ at positions $l'\ge l$.
Every tuple produced thus keeps the prefix of $u$, and the closure never leaves $\Fix{R}{p}$.
This is what allows us to explore $\Fix{R}{p}$ without first computing a representation of it, which is the task of \textsc{Fix-values} in~\cite{BD06,Dal04} and its analogues~\cite{IMMVW10,Kom24}.

\begin{lemma}\label{lem:fibre}
Let $R\subseteq D^{n}$ be a Maltsev closed relation with a representation $W$. 
Let $l\in[n]$, $p\in D^{l-1}$ and $u\in \Fix{R}{p}$.
Then $\Fix{R}{p}=\Cl_{F(W_{\ge l})}(u)$.
\end{lemma}
\begin{proof}
For ``$\supseteq$'': it suffices to show that $\Fix{R}{p}$ is closed under every $f\in F(W_{\ge l})$, since $u\in \Fix{R}{p}$.
Let $x\in \Fix{R}{p}$ and $(l',a,a')\in W_{\ge l}$.
By Lemma~\ref{lem:move} $f_{a,a'}(x)\in R$ and $f_{a,a'}(x)_{i}=x_{i}=p_{i}$ for all $i<l\le l'$. Hence, $f_{a,a'}(x)\in \Fix{R}{p}$.

For ``$\subseteq$'': let $t\in \Fix{R}{p}$.
We construct tuples $x^{(l-1)},x^{(l)},\dots,x^{(n)}$ such that, for each $l'$,
\[
  x^{(l')}\in\Cl_{F(W_{\ge l})}(u)\quad\text{and}\quad x^{(l')}\text{ agrees with }t\text{ on the first }l'\text{ coordinates.}
\]
For $l'=l-1$ we simply take $x^{(l-1)}=u$: it trivially lies in the closure and agrees with $t$ on the first $l-1$ coordinates because both $u$ and $t$ have prefix $p$.
For the step from $l'-1$ to $l'$ (where $l\le l'\le n$), the tuples $x^{(l'-1)}$ and $t$ both lie in $R$ (the closure is contained in $\Fix{R}{p}\subseteq R$ by ``$\supseteq$'') and agree on $[{<}l']$, so the pair $(x^{(l'-1)},t)$ witnesses $(l',x^{(l'-1)}[l'],t[l'])\in\Sig(R)$.
Hence $W$ contains a triple $(l',a,a')$ with $a[l']=x^{(l'-1)}[l']$ and $a'[l']=t[l']$, and $f_{a,a'}\in F(W_{\ge l})$ because $l'\ge l$.
Put $x^{(l')}=f_{a,a'}(x^{(l'-1)})$, which lies in the closure. By Lemma~\ref{lem:move} it agrees with $t$ on $[{<}l']$ and has $x^{(l')}[l']=t[l']$, i.e.\ it agrees with $t$ on the first $l'$ coordinates.
Finally, $x^{(n)}$ agrees with $t$ everywhere, so $t=x^{(n)}\in\Cl_{F(W_{\ge l})}(u)$.
\end{proof}

In particular we obtain the following corollary (by letting $l = 1$ and $p$ be the empty tuple).

\begin{corollary}\label{cor:generate}
If $W$ is a representation of a Maltsev closed $R \neq \emptyset$ and $w \in R$, then $R = \Cl_{F(W)}(w)$.
\end{corollary}

However, the forthcoming algorithm never applies moves to full $n$-tuples when searching. Instead we apply the projected moves $f^I_{a,a'}$ to projections onto a short index sequence $I$ (decided by one position $l$ and the scope of a constraint). The next, trivial, lemma says that this can be done without loss of generality.
Hence, the closures of Lemma~\ref{lem:fibre} can be computed on $\pr_I(R)$ instead of on $R$.

\begin{lemma} \label{lem:proj}
Let $I$ be an index sequence of length $s$ and let $\pi=\pr_{I}$.
For all $a,a'\in D^{n}$,
\[
  \pi\circ f_{a,a'}=f_{a_{I},a'_{I}}\circ\pi .
\]
Consequently, for $u\in D^{n}$ and $a^{1},a'^{1},\dots,a^{r},a'^{r}\in D^{n}$,
\begin{equation*}
  \pi\bigl(f_{a^{r},a'^{r}}\circ\dots\circ f_{a^{1},a'^{1}}(u)\bigr)=f_{a^{r}_{I},a'^{r}_{I}} \circ \dots \circ f_{a^{1}_{I},a'^{1}_{I}}(\pi(u)),
\end{equation*}
and for every list $V$ of triples, $\pi(\Cl_{F(V)}(u))=\Cl_{F_{I}(V)}(\pi(u))$.
\end{lemma}
\begin{proof}
For any $x \in D^n$ we have $\pi(\varphi(x,a,a'))=\varphi(\pi(x),\pi(a),\pi(a'))$.
The second equation follows by induction on $r$.
\end{proof}

The last lemma that we need is a simple statement involving repeated variables: if we have a $k$-ary constraint over $k' < k$ variables (i.e., some variables are repeated) then we can construct an equivalent, $k'$-ary relation, that is also preserved by Maltsev. We omit the proof since it is trivial.

\begin{lemma} \label{lem:repeat}
Let $Q\in[n]^{k}$, let $S\subseteq D^{k}$ be Maltsev closed, let $Q'\in[n]^{k'}$ list the distinct entries of $Q$ in order of first occurrence, and let $S'=\pr_{Q'}(C_{Q,S})$.
Then $S'$ is Maltsev closed, $C_{Q',S'}=C_{Q,S}$, and $|S'|\le|S|$.
Explicitly, $S'$ consists of the tuples $s\in S$ with $s_{i}=s_{i'}$ whenever $Q_{i}=Q_{i'}$, each with its repeated positions deleted, and it is computed from $S$ in time $O(k|S|)$.
\end{lemma}

\subsection{The algorithm}\label{sec:algorithm}

We describe the algorithm in a style similar to the simple algorithm for Maltsev constraints (referred to as \textsc{BD}) by Bulatov \& Dalmau~\cite{BD06}. See also Dyer \& Richerby~\cite{DR13} from which some of our notation is borrowed.
Throughout this section, $D$ is a finite set with $d$ elements, $\varphi$ is, as usual, a Maltsev operation on $D$, and $n\ge1$ an integer.
A Maltsev CSP instance $\Inst$ then consists of $n$ variables $x_{1},\dots,x_{n}$ and constraints $(Q_{1},S_{1}),\dots,(Q_{m},S_{m})$ with $Q_{j}\in[n]^{k_{j}}$ (repetitions allowed) where each $S_{j}\subseteq D^{k_{j}}$ is Maltsev closed and represented by a list of tuples. The solutions to this instance are the $t\in D^{n}$ with $\pr_{Q_j}(t)\in S_{j}$ for all $j$.

The algorithm uses the same idea as \textsc{BD} and maintains a compact representation $W$ of the set $R$ of assignments satisfying the constraints treated so far.
Initially, $R=D^{n}$, represented by the $nd^{2}$ triples $(l,e_{l,b},e_{l,c})$, where $e_{l,b}$ has $b$ at position $l$ and a fixed value $d_{0}$ elsewhere.
When we pick a constraint $(Q,S)$ then we thus need to produce a compact representation of $T=R\cap C_{Q,S}$ from $W$, $Q$ and $S$ alone. 
A compact representation of $T$ then needs, for every $(l,b,c)\in\Sig(T)$, two tuples of $T$ with a common prefix before $l$ and values $b,c$ at $l$. We can accomplish this as follows.

\begin{enumerate}
    \item 
    Let $1 \leq l \leq n$ and $b \in D$. We first attempt to compute a tuple $u_{l,b} \in T$ such that $u_{l,b}[l] = b$ (we describe how to do this in the last item of this list). We will use this tuple to describe all possible forks of the form $(l,b,c)$ for $c \in D$.
    \item 
    If we let $p = u_{l,b}[{<}l]$ then we now want to describe the set 
    $\{c \mid (l,b,c)\in\Sig(T)\}$, which by Lemma~\ref{lem:transport} equals $\{t'[l] \mid t'\in T,\ t'[{<}l]=p\} = \pr_{l}(\Fix{T}{p})$.
    \item 

    Since $\Fix{T}{p}=\Fix{R}{p}\cap C_{Q,S}$, it is enough to generate $\Fix{R}{p}$ and keep the tuples whose projection onto $Q$ lies in $S$.
By Lemma~\ref{lem:fibre}, $\Fix{R}{p}$ is obtained from the single tuple $u_{l,b}$ by repeatedly picking a tuple $x$ generated so far (initially only $u_{l,b}$) and a triple $(l',a,a')\in W$ with $l'\ge l$, and forming $\varphi(x,a,a')$. Since $a$ and $a'$ agree on the first $l'-1 \geq l-1$ positions, this tuple again has prefix $p$.
Only the triples of $W$ at positions $l'\ge l$ are  thus used here.
     \item 
Every tuple $\varphi(x,a,a')$ describes a possibly new element of $\pr_{\lQ}(\Fix{R}{p})$, which in turn, if its projection onto $Q$ lies in $S$, can be used to describe a fork with respect to $u_{l,b}$. However, in this closure step, we only need to remember the values at position $l$ and $Q$, and thus apply $\varphi$ to the projection of $x,a,a'$ to $l$ and $Q$. One full $n$-tuple per new projected tuple is carried along (said to be ``lifted''), since the witness needs to be an $n$-tuple and not a projected tuple.
     \item 
     To compute the tuples of the form $u_{l,b}$ we pick an arbitrary tuple $w \in \Tup(W)$ and then explore the tuples reachable from $w$ with respect to $(l', a, a') \in W$ (and any new tuples obtained in this way). When computing these Maltsev applications we, again, only need to worry about the values at index $l$ and $Q$ since we after this search need to remove the tuples inconsistent with $S$. However, since $u_{l,b} \in D^n$, we lift each projected tuple (over $l$ and $Q)$ to a full $n$-ary tuple.
\end{enumerate}

Similar to \textsc{BD} we avoid computing the above for $(Q,S)$ directly and instead do it gradually over one more additional position in the constraint scope in each iteration. In this way we can bound the number of projected tuples to $\le d^2|S|$ rather than $d^{|Q|+1}$.
However, in contrast, \textsc{BD} compute  a full compact representation of $\Fix{R}{p}$ (for all $nd$ pairs $(l,b)$) with the procedure \textsc{Fix-values}. This requires $\Theta(nd^{2})$ tuples of length $n$, hence $\Omega(n^{2})$ per pair and $\Omega(n^{3})$ per coordinate step, regardless of how it is implemented ($n^{4}$ with \textsc{Fix-values} from Kompatscher~\cite{Kom24}, but $n^{8}$ in the analysis of~\cite{BD06}, where \textsc{Nonempty} closes sets of full tuples under the ternary operation).
In contrast, we never compute all of $\Fix{R}{p}$. Step~4 computes its projection onto the $|Q|+1$ positions $(l,Q)$, i.e.\ at most $d^{2}|S|$ projected tuples, and each one is obtained by one pass over the at most $nd^{2}$ triples of $W$, plus one $O(n)$ to  ``lift'' a projected tuple to an $n$-ary tuple.
This is $O(nd^{4}|S|)$ per pair, $O(n^{2}d^{5}|S|)$ per coordinate step, and $O(k\,n^{2}d^{5}|S|)$ for a constraint of arity $k$. The two factors $n$ are thus the $n$ possible positions $l$, and the $\Theta(nd^{2})$ size of the compact representation $W$.

With this description in mind we now turn to describing the different procedures in detail.

\paragraph{\Reach.}
Given an index sequence $I$ of length $s$, a list $V$ of triples and $u\in D^{n}$, \Reach\ performs a breadth-first search from $\pr_{I}(u)$ in the directed graph on $D^{s}$ whose edges are the projected moves $f^{I}_{a,a'}$, $(l',a,a')\in V$, and \emph{lifts} every vertex $g$ it discovers to a full tuple $x_{g}\in D^{n}$ with $\pr_{I}(x_{g})=g$: whenever the projected move of $(a,a')$ leads from $g$ to a new vertex $g'$, the move $f_{a,a'}$ itself is applied to $x_{g}$ to obtain $x_{g'}$.
It plays the role of \textsc{Nonempty} in \cite{BD06}, but runs on projections.

\begin{procedure}{Reach}{I,V,u}
\textit{Input:} an index sequence $I$ of length $s$, a list $V$ of triples $(l',a,a')$ with $a,a'\in D^{n}$, a tuple $u\in D^{n}$.\\
\textit{Output:} a set $G\subseteq D^{s}$ and, for each $g\in G$, a tuple $x_{g}\in D^{n}$ with $\pr_{I}(x_{g})=g$, obtained from $u$ by applying moves $f_{a,a'}$ with $(l',a,a')\in V$ (Lemma~\ref{lem:reach}).
\begin{steps}
\item $G\gets\{\pr_{I}(u)\}$;\ $x_{\pr_{I}(u)}\gets u$;\ $\mathit{Queue}\gets(\pr_{I}(u))$.
\item \textbf{while} $\mathit{Queue}\neq\emptyset$ \textbf{do}
  \begin{steps}
  \item remove the first element $g$ of $\mathit{Queue}$;
  \item \textbf{for each} $(l',a,a')\in V$ \textbf{do}: $g'\gets\varphi(g,\pr_{I}(a),\pr_{I}(a'))$; \textbf{if} $g'\notin G$ \textbf{then} $x_{g'}\gets\varphi(x_{g},a,a')$, add $g'$ to $G$ and to the end of $\mathit{Queue}$.
  \end{steps}
\item \textbf{return} $(G,(x_{g})_{g\in G})$.
\end{steps}
\end{procedure}

\paragraph{\NextOne.}
Given a representation $W$ of a closed $R$ and a constraint $(Q,S)$, \NextOne\ computes a compact representation of $T=R\cap C_{Q,S}$, or reports that $T=\emptyset$.
It corresponds to \textsc{Next-$\beta$} in \textsc{BD} but the calls to \textsc{Fix-values} are replaced by calls to \Reach\ with the list $W_{\ge l}$ (Lemma~\ref{lem:fibre}).

\begin{procedure}{Next-one}{W,Q,S}
\textit{Input:} a representation $W$ of a Maltsev closed relation $R\subseteq D^{n}$, $Q\in[n]^{k}$, a Maltsev closed relation $S\subseteq D^{k}$.\\
\textit{Output:} a compact representation of $R\cap C_{Q,S}$, or ``empty''.
\begin{steps}
\item \textbf{if} $W=\emptyset$ \textbf{then return} ``empty''; otherwise pick any $w\in\Tup(W)$.
\item \textbf{for each} $l\in[n]$ \textbf{do}
  \begin{steps}
  \item $(G,x)\gets\Reach(\lQ,W,w)$;
  \item \textbf{for each} $b\in D$ \textbf{do}: \textbf{if} there is $g\in G$ with $g_{1}=b$ and $g_{>1}\in S$ \textbf{then} $u_{l,b}\gets x_{g}$ for one such $g$, \textbf{else} $u_{l,b}\gets\bot$.
  \end{steps}
\item \textbf{if} $u_{1,b}=\bot$ for all $b\in D$ \textbf{then return} ``empty''.
\item $W'\gets\emptyset$; \textbf{for each} $(l,b)\in[n]\times D$ with $u_{l,b}\neq\bot$ \textbf{do}
  \begin{steps}
  \item $(G,x)\gets\Reach(\lQ,W_{\ge l},u_{l,b})$;
  \item \textbf{for each} $c\in D$ \textbf{do}: \textbf{if} there is $g\in G$ with $g_{1}=c$ and $g_{>1}\in S$ \textbf{then} choose one such $g$ and add $(l,u_{l,b},x_{g})$ to $W'$.
  \end{steps}
\item \textbf{return} $W'$.
\end{steps}
\end{procedure}

\paragraph{\Next.}
A constraint is processed one scope coordinate at a time (as in \textsc{BD}).
This is what keeps every set $G$ computed by \Reach\ small (Lemma~\ref{lem:small}).

\begin{procedure}{Next}{W,Q,S}
\textit{Input:} a representation $W$ of a Maltsev closed relation $R\subseteq D^{n}$, a constraint $(Q,S)$, $Q\in[n]^{k}$, where $S\subseteq D^{k}$ is Maltsev closed.\\
\textit{Output:} a compact representation of $R\cap C_{Q,S}$, or ``empty''.
\begin{steps}
\item Compute $Q'=(q'_{1},\dots,q'_{k'})$ and $S'$ from $(Q,S)$ as in Lemma~\ref{lem:repeat}; \textbf{if} $S'=\emptyset$ \textbf{then return} ``empty''; \textbf{if} $k'=0$ \textbf{then return} $W$.
\item \textbf{for} $j=1,\dots,k'$ \textbf{do}
  \begin{steps}
  \item $W\gets\NextOne\bigl(W,\ Q'_{\le j},\ \pr_{\le j}(S')\bigr)$;
  \item \textbf{if} the result is ``empty'' \textbf{then return} ``empty''.
  \end{steps}
\item \textbf{return} $W$.
\end{steps}
\end{procedure}

\paragraph{\Solve.}
Fix $d_{0}\in D$ and let $e_{l,b}\in D^{n}$ be the tuple with $b$ at position $l$ and $d_{0}$ elsewhere.

\begin{procedure}{Solve}{\Inst}
\textit{Input:} an instance $\Inst$ with variables $x_{1},\dots,x_{n}$ and constraints $(Q_{1},S_{1}),\dots,(Q_{m},S_{m})$.\\
\textit{Output:} ``unsatisfiable'', or a solution and a compact representation of the set of solutions.
\begin{steps}
\item $W\gets$ the list of all $(l,e_{l,b},e_{l,c})$, $l\in[n]$, $b,c\in D$.
\item \textbf{for} $j=1,\dots,m$ \textbf{do}: $W\gets\Next(W,Q_{j},S_{j})$; \textbf{if} the result is ``empty'' \textbf{then return} ``unsatisfiable''.
\item \textbf{return} any tuple of $\Tup(W)$, and $W$.
\end{steps}
\end{procedure}

\section{Correctness and Complexity}\label{sec:correctness}

The correctness of the algorithm in Section~\ref{sec:algorithm} follows by similar arguments to Bulatov \& Dalmau~\cite{BD06} and the basic lemmas established in Section~\ref{sec:lemmas}. However, for completeness sake we analyze each procedure in detail, starting with \Reach.

\begin{lemma}\label{lem:reach}
Let $(G,x)$ be the output of $\Reach(I,V,u)$.
Then $G=\Cl_{F_{I}(V)}(\pr_{I}(u))$, and $x_{g}\in\Cl_{F(V)}(u)$ and $\pr_{I}(x_{g})=g$ for every $g\in G$.
\end{lemma}
\begin{proof}
The search starts with $\pr_{I}(u)$ and then explores the graph whose edges from $g$ lead to the vertices of the form $f^{I}_{a,a'}(g)$, for $(l',a,a')\in V$. Thus, $G$ is the set of vertices reachable from $\pr_{I}(u)$, i.e.\ $\Cl_{F_{I}(V)}(\pr_{I}(u))$.
For the second claim we induct on the order in which vertices are added.
Initially, $x_{\pr_{I}(u)}=u\in\Cl_{F(V)}(u)$ and $\pr_{I}(x_{\pr_{I}(u)})=\pr_{I}(u)$.
When $g'=f^{I}_{a,a'}(g)$ is added, we set $x_{g'}=f_{a,a'}(x_{g})$. This is included in $\Cl_{F(V)}(u)$ since $x_{g}$ is, and $f_{a,a'}\in F(V)$. By Lemma~\ref{lem:proj}, $\pr_{I}(x_{g'})=f^{I}_{a,a'}(\pr_{I}(x_{g}))=f^{I}_{a,a'}(g)=g'$.
\end{proof}

We continue with \NextOne.

\begin{theorem} \label{thm:nextone}
Let $R\subseteq D^{n}$ be Maltsev closed, let $W$ be a representation of $R$, let $Q\in[n]^{k}$, let $S\subseteq D^{k}$ be Maltsev closed, and let $T=R\cap C_{Q,S}$.
Then:
\begin{enumerate}
\item $\NextOne(W,Q,S)$ returns ``empty'' if and only if $T=\emptyset$, and
\item otherwise it returns a compact representation $W'$ of $T$.
\end{enumerate}
\end{theorem}

\begin{proof}
First, $T$ is clearly Maltsev closed. We first deal with the simple case, then describe what the two kinds of calls of \Reach\ compute, and, finally, verify the two claims of the theorem.

If  $R\ne\emptyset$ then $\Sig(R)\ne\emptyset$, so every representation of $R$ is non-empty.
Hence if $W=\emptyset$ then $R=\emptyset$, so $T=\emptyset$ and ``empty'' is the correct answer.
Assume from now on $W\ne\emptyset$, so that a tuple $w\in\Tup(W)$ can be chosen in Step~1 of \NextOne. Note that $w\in R$ since $\Tup(W)\subseteq R$.

We continue with describing the searches from $w$.
Fix $l\in[n]$.
By Corollary~\ref{cor:generate}, $R=\Cl_{F(W)}(w)$, and projecting with Lemma~\ref{lem:proj} gives
\[
  \pr_{\lQ}(R)=\Cl_{F_{\lQ}(W)}(\pr_{\lQ}(w)).
\]
So by Lemma~\ref{lem:reach} the call $\Reach(\lQ,W,w)$ in Step~2.1 returns
\begin{equation}\label{eq:searchR}
  G=\pr_{\lQ}(R),\qquad\text{with } x_{g}\in R \text{ and } \pr_{\lQ}(x_{g})=g \text{ for every } g\in G .
\end{equation}
Consequently, the test in Step~2.2 succeeds for $b$ if and only if some $t\in R$ has $t_{l}=b$ and $\pr_{Q}(t)\in S$. I.e., if and only if $b\in\pr_{l}(T)$, and when it succeeds, the chosen $u_{l,b}=x_{g}$ lies in $R$, has $\pr_{Q}(u_{l,b})=g_{>1}\in S$, hence lies in $T$, and has $u_{l,b}[l]=g_{1}=b$.

Next, we describe the searches from $u_{l,b}$.
Fix some $(l,b)$ with $u=u_{l,b}\ne\bot$ and let $p=u[{<}l]$.
By the argument in the preceding paragraph, $u\in T\subseteq R$, so $u\in\Fix{R}{p}$, and Lemma~\ref{lem:fibre} gives $\Fix{R}{p}=\Cl_{F(W_{\ge l})}(u)$, so by projecting with Lemma~\ref{lem:proj} and applying Lemma~\ref{lem:reach}, the call $\Reach(\lQ,W_{\ge l},u)$ in Step~4.1 returns
\begin{equation}\label{eq:searchFix}
  G=\pr_{\lQ}(\Fix{R}{p}),\qquad\text{with } x_{g}\in\Fix{R}{p} \text{ and } \pr_{\lQ}(x_{g})=g \text{ for every } g\in G .
\end{equation}
Hence, for $c\in D$, the test in Step~4.2 succeeds if and only if some $t\in\Fix{R}{p}$ has $t_{l}=c$ and $\pr_{Q}(t)\in S$, i.e.\ if and only if $c\in\pr_{l}(\Fix{R}{p}\cap C_{Q,S})=\pr_{l}(\Fix{T}{p})$ (the constraint $C_{Q,S}$ does not involve the prefix, so $\Fix{T}{p}=\Fix{R}{p}\cap C_{Q,S}$).
By Lemma~\ref{lem:transport}, applied to the closed set $T$ and to $u\in T$ with $u[l]=b$, this is the case if and only if $(l,b,c)\in\Sig(T)$.
When the test succeeds, the triple $(l,u,x_{g})$ added to $W'$ satisfies the following.
Since $x_{g}\in\Fix{R}{p}$, we have $x_{g}\in R$ and $x_{g}[{<}l]=p=u[{<}l]$.
Moreover, $x_{g}[l]=g_{1}=c$ and $\pr_{Q}(x_{g})=g_{>1}\in S$, so $x_{g}\in T$.
Together with $u\in T$, $u[l]=b$: every triple in $W'$ consists of two tuples of $T$ with a common prefix, witnessing $(l,b,c)\in\Sig(T)$, and for each $(l,b)$ at most one triple is added per $c$.

We now wrap up as follows.
For the first claim, $T=\emptyset$ if and only if $\pr_{1}(T)=\emptyset$, which holds if and only if $u_{1,b}=\bot$ for all $b$, which is exactly when Step~3 returns ``empty''.
For the second claim, suppose $T\ne\emptyset$ and let $(l,b,c)\in\Sig(T)$.
Then $b\in\pr_{l}(T)$, so $u_{l,b}\ne\bot$ and by the previous paragraph Step~4.2 adds a triple $(l,u_{l,b},x_{g})$ with $u_{l,b}[l]=b$ and $x_{g}[l]=c$.
So every element of $\Sig(T)$ is witnessed in $W'$, every triple of $W'$ witnesses an element of $\Sig(T)$, and no element is witnessed twice. Hence, $W'$ is a compact representation of $T$.
\end{proof}

We work our way through the procedures and prove correctness of \Next.

\begin{theorem} \label{thm:next}
Let $R\subseteq D^{n}$ be Maltsev closed, $W$ a representation of $R$, and $(Q,S)$ a constraint where $S$ is Maltsev closed.
Then $\Next(W,Q,S)$ returns ``empty'' if and only if $R\cap C_{Q,S}=\emptyset$, and otherwise returns a compact representation of $R\cap C_{Q,S}$.
\end{theorem}
\begin{proof}
Let $Q',S'$ be as in Lemma~\ref{lem:repeat}. Then $S'$ is Maltsev closed and $C_{Q',S'}=C_{Q,S}$.
Consider first the early exits in Step~1.
If $S'=\emptyset$ then $C_{Q,S}=\emptyset$ and ``empty'' is correct.
If $k'=0$ then $Q$ is empty and $S'=\{\varepsilon\}$, so $C_{Q,S}=D^{n}$, $R\cap C_{Q,S}=R$, and returning $W$ is correct.

Thus, assume $S'\ne\emptyset$ and $k'\ge1$, and for $j=0,\dots,k'$ let
\[
  S'_{j}=\pr_{\le j}(S')\qquad\text{and}\qquad R^{(j)}=R\cap C_{Q'_{\le j},S'_{j}} .
\]
Then $R^{(0)}$ is simply $R$.
Further, $R^{(k')}=R\cap C_{Q',S'}=R\cap C_{Q,S}$, each $S'_{j}$ is Maltsev closed, and $R^{(j)}=R^{(j-1)}\cap C_{Q'_{\le j},S'_{j}}$ because $C_{Q'_{\le j},S'_{j}}\subseteq C_{Q'_{\le j-1},S'_{j-1}}$.
We claim that whenever iteration $j$ of Step~2 starts, $W$ is a representation of $R^{(j-1)}$, and that iteration $j$ returns ``empty'' if and only if $R^{(j)}=\emptyset$.
For $j=1$, $W$ represents $R^{(0)}=R$ by assumption.
If $W$ represents $R^{(j-1)}$ when iteration $j$ starts, then Theorem~\ref{thm:nextone}, applied to $R^{(j-1)}$ and the constraint $(Q'_{\le j},S'_{j})$, says that iteration $j$ returns ``empty'' if and only if $R^{(j-1)}\cap C_{Q'_{\le j},S'_{j}}=R^{(j)}$ is empty, and otherwise replaces $W$ by a compact representation of $R^{(j)}$, from which iteration $j+1$ starts.
This proves the claim.
Since $R^{(j)}\supseteq R^{(k')}=R\cap C_{Q,S}$ for every $j$, an answer ``empty'' in any iteration is correct, and if no iteration returns ``empty'' then the final $W$ is a compact representation of $R^{(k')}=R\cap C_{Q,S}$.
\end{proof}

Finally, we conclude the correctness proof with \Solve.

\begin{theorem} \label{thm:solve}
\Solve\ returns ``unsatisfiable'' if and only if the instance has no solution.
Otherwise the returned tuple is a solution and the returned $W$ is a compact representation of the set of all solutions.
\end{theorem}
\begin{proof}
For $j\in\{0,\dots,m\}$ let $\mathrm{Sol}_{j}=\bigcap_{i\le j}C_{Q_{i},S_{i}}$ be the set of assignments satisfying the first $j$ constraints. So, $\mathrm{Sol}_{0}=D^{n}$, $\mathrm{Sol}_{j}=\mathrm{Sol}_{j-1}\cap C_{Q_{j},S_{j}}$, each $\mathrm{Sol}_{j}$ is Maltsev closed, and $\mathrm{Sol}_{m}$ is the set of solutions.
The initial list is a compact representation of $\mathrm{Sol}_{0}=D^{n}$ (the same one as in \cite{BD06}): $\Sig(D^{n})=[n]\times D\times D$, and for each $(l,b,c)$ the triple $(l,e_{l,b},e_{l,c})$ is its unique witness, as $e_{l,b}$ and $e_{l,c}$ agree before position $l$.
By induction on $j$, using Theorem~\ref{thm:next} for the step, as long as \Solve\ has not stopped, after iteration $j$ of Step~2, $W$ is a compact representation of $\mathrm{Sol}_{j}$, and iteration $j$ returns ``empty'' if and only if $\mathrm{Sol}_{j}=\emptyset$.
If some iteration $j$ returns ``empty'' then $\mathrm{Sol}_{m}\subseteq \mathrm{Sol}_{j}=\emptyset$ and the answer is correct.
Otherwise $\mathrm{Sol}_{m}\ne\emptyset$, and $W$ is a compact representation of $\mathrm{Sol}_{m}$.
%so $W\ne\emptyset$  and $\Tup(W)\subseteq \mathrm{Sol}_{m}$ consists of solutions.
\end{proof}

For the complexity, throughout, we have $|W|\le nd^{2}$ since we in each iteration produce a compact representation. 
The total cost of \Solve\ is dominated by the calls of \Reach, so we first bound the sets they explore.

\begin{lemma}\label{lem:small}
Let $R\subseteq D^{n}$ be Maltsev closed, $W$ a representation of $R$, $k\ge1$, $Q=(q_{1},\dots,q_{k})\in[n]^{k}$, $S\subseteq D^{k}$ closed, and let $S^{-}\subseteq D^{k-1}$ satisfy $\pr_{Q_{\le k-1}}(R)\subseteq S^{-}$.
Then in every call of \Reach\ made by $\NextOne(W,Q,S)$ the returned set satisfies
\[
  G\subseteq D\times S^{-}\times D,\qquad\text{so}\qquad |G|\le d^{2}\,|S^{-}|.
\]
\end{lemma}
\begin{proof}
By~\eqref{eq:searchR} and~\eqref{eq:searchFix} in the proof of Theorem~\ref{thm:nextone}, $G=\pr_{\lQ}(R)$ in Step~2.1 and $G=\pr_{\lQ}(\Fix{R}{p})\subseteq\pr_{\lQ}(R)$ in Step~4.1.
Since $\lQ=(l,Q_{\le k-1},q_{k})$,
\[
  \pr_{\lQ}(R)=\bigl\{(t_{l},\pr_{Q_{\le k-1}}(t),t_{q_{k}}) \mid t\in R\bigr\}\subseteq D\times S^{-}\times D 
\]
which concludes the proof.
\end{proof}

\begin{lemma}\label{prop:cost}
A call $\Next(W,Q,S)$ with $Q\in[n]^{k}$ costs $O\bigl(d^{5}\,n^{2}\,k|S|\bigr)$, and \Solve\ runs in time $O(d^{5}\,n^{2}\,\size{\Inst})$.
\end{lemma}
\begin{proof}
Step~1 of \Next\ costs $O(k|S|)$ (Lemma~\ref{lem:repeat}).
Fix iteration $j$ of Step~2, i.e.\ the call $\NextOne(W,Q'_{\le j},S'_{j})$, and write $S^{-}$ for $S'_{j-1}$.
Every vertex $g$ handled by \Reach\ is a tuple in $D^{j+1}$, namely $g=\pr_{\lQ}(t)$ for some $t\in R$.
Write it as $g=(g_{1},s,g_{j+1})$ with $s=(g_{2},\dots,g_{j})\in D^{j-1}$.
By Lemma~\ref{lem:small}, $s\in S^{-}$.
Fix an enumeration $s^{1},\dots,s^{|S^{-}|}$ of $S^{-}$.
Then $g$ is determined by the three numbers $g_{1}\in D$, the index $i$ with $s=s^{i}$, and $g_{j+1}\in D$, and we store $g$ as this triple of numbers rather than as a $(j+1)$-tuple.
Consequently the set $G$ can be kept as an array of $d\cdot|S^{-}|\cdot d$ bits, so membership in $G$ is tested in constant time.
Likewise, whether $g_{>1}=(s^{i},g_{j+1})$ lies in $S'_{j}$ is read off a table with one bit per pair $(i,g_{j+1})$, which is filled once per call of \NextOne\ by running through $S'_{j}$, in $O(j|S'_{j}|)$ time.
Now consider a projected move $f^{\lQ}_{a,a'}$ applied to a vertex $(g_{1},s,g_{j+1})$.
It changes $g_{1}$ to $\varphi(g_{1},a_{l},a'_{l})$ and $g_{j+1}$ to $\varphi(g_{j+1},a_{q'_{j}},a'_{q'_{j}})$; this is two applications of $\varphi$.
It changes $s=s^{i}$ to $\varphi(s^{i},\pr_{Q'_{\le j-1}}(a),\pr_{Q'_{\le j-1}}(a'))$.
This new middle part again lies in $S^{-}$, since $a,a'\in R$, $\pr_{Q'_{\le j-1}}(R)\subseteq S^{-}$, and $S^{-}$ is closed.
It depends only on $i$ and on the triple $(l',a,a')\in W$.
We therefore compute it once for every $i$ and every triple of $W$, and store its index in $S^{-}$ in a table.
Each entry costs $O(j)$, so the table costs $O(j\,|S^{-}|\,|W|)=O(j\,n\,d^{2}|S|)$ per call of \NextOne.
With the table, every projected move is applied in $O(1)$.

One call of \Reach\ therefore inserts at most $d^{2}|S|$ vertices, spends $O(n)$ per inserted vertex on the lifted tuple $\varphi(x_{g},a,a')$, and $O(|W|)=O(nd^{2})$ per dequeued vertex on the scan of the list, which is $O(d^{2}|S|\,(n+nd^{2}))=O(nd^{4}|S|)$.
\NextOne\ makes $n$ such calls in Step~2 and at most $nd$ in Step~4, so one coordinate step costs $O(n^{2}d^{5}|S|)$, which absorbs the table.
Summing over the $k'\le k$ coordinate steps gives the bound for \Next, and summing $O(k_{j}n^{2}d^{5}|S_{j}|)$ over the constraints gives $O(d^{5}n^{2}\size{\Inst})$ for \Solve, where the $O(n^{2}d^{2})$ initial list being absorbed since $n\le\size{\Inst}$.
\end{proof}

We combine the results in this sections and conclude our main result.

\begin{theorem}\label{thm:main}
Let $\varphi \colon D^3 \to D$ be a Maltsev operation. Then an instance $\Inst$ of $\CSP(\Inv(\varphi))$ with $n$ variables is solvable in  $O(d^{5}\,n^{2}\,\size{\Inst})$ time, where $\size{\Inst}$ is the size of $\Inst$. 
\end{theorem}

\bibliographystyle{abbrv}
\bibliography{bib}

\end{document}